\documentclass[11pt,a4paper]{article}
\usepackage{etex}
\usepackage[utf8]{inputenc}
\usepackage{amsmath}
\usepackage{amsfonts}
\usepackage{amssymb}
\usepackage{parskip}
\usepackage{amsthm}
\usepackage{thmtools}
\usepackage{etoolbox}
\usepackage{cite}
\usepackage{color}

\begingroup
    \makeatletter
    \@for\theoremstyle:=definition,remark,plain\do{%
        \expandafter\g@addto@macro\csname th@\theoremstyle\endcsname{%
            \addtolength\thm@preskip\parskip
            }%
        }
\endgroup

\usepackage{mathtools}
\usepackage[left=3cm,right=3cm,top=3cm,bottom=3cm]{geometry}
\usepackage{array}
\usepackage[all,cmtip]{xy}
\usepackage{graphicx}
\usepackage{tikz,pgfplots}
\pgfplotsset{compat=1.10}

\usepackage{float}
\usepackage{subcaption}
\usepackage{todonotes}

\usepackage{booktabs}
\usepackage{ifdraft}

\usepackage{microtype}	
\usepackage[colorlinks,citecolor=black,urlcolor=black]{hyperref}
\usepackage[nameinlink]{cleveref}

\numberwithin{equation}{section}
\numberwithin{table}{section}
\numberwithin{figure}{section}

\newtheorem{theorem}{Theorem}[section]
\crefname{theorem}{theorem}{theorems}
\Crefname{theorem}{Theorem}{Theorems}

\crefname{prop}{proposition}{propositions}
\Crefname{prop}{Proposition}{Propositions}

\crefname{lemma}{lemma}{lemmas}
\Crefname{lemma}{Lemma}{Lemmas}

\newtheorem{corollary}[theorem]{Corollary}
\crefname{corollary}{corollary}{corollaries}
\Crefname{corollary}{Corollary}{Corollaries}

\theoremstyle{definition}
\newtheorem{remark}[theorem]{Remark}
\AtEndEnvironment{remark}{\hfill\ensuremath{\diamondsuit}}
\crefname{remark}{remark}{remarks}
\Crefname{remark}{Remark}{Remarks}

\theoremstyle{definition}

\crefname{defn}{definition}{definitions}
\Crefname{defn}{Definition}{Definitions}

\AtEndEnvironment{example}{\hfill\ensuremath{\bigcirc}}
\crefname{example}{example}{example}
\Crefname{example}{Example}{Examples}

\DeclareMathOperator*{\spann}{span}
\DeclareMathOperator{\id}{Id}

\newcommand{\bff}{\boldsymbol{f}}
\newcommand{\bfg}{\boldsymbol{g}}

\newcommand{\system}{\boldsymbol{\Sigma}}

\newcommand{\stateSpace}{\mathcal{Z}}
\newcommand{\state}{\boldsymbol{z}}
\newcommand{\stateNorm}[1]{\|#1\|_{\stateSpace}}

\newcommand{\inpSpace}{\mathcal{W}}
\newcommand{\inpSet}{\mathbb{W}}
\newcommand{\inpVar}{\boldsymbol{w}}
\newcommand{\inpNorm}[1]{{\|#1\|_{\inpSpace}}}
\newcommand{\inpProd}[2]{\left\langle #1,#2\right\rangle_{\inpSpace}}
\newcommand{\inpSpaceRed}{\inpSpace_n}
\newcommand{\inpSpaceRedMin}{\smash{\widehat{\inpSpace}_n}}

\newcommand{\outSpace}{\mathcal{Y}}
\newcommand{\outVar}{\boldsymbol{y}}
\newcommand{\redOutVar}{\outVar_n}
\newcommand{\outNorm}[1]{\left\|#1\right\|_{\outSpace}}
\newcommand{\outNormSmashed}[1]{\|#1\|_{\outSpace}}
\newcommand{\outProd}[2]{{\left\langle#1,#2\right\rangle_{\outSpace}}}
\newcommand{\outSpaceRed}{\outSpace_n}
\newcommand{\outSpaceRedMin}{\widehat{\outSpace}_n}

\newcommand{\stateDim}{N}
\newcommand{\controlDim}{m}
\newcommand{\outputDim}{p}

\newcommand{\param}{\boldsymbol{p}}
\newcommand{\control}{\boldsymbol{u}}
\newcommand{\paramSpace}{\mathcal{P}}
\newcommand{\paramSet}{\mathbb{P}}
\newcommand{\controlSpace}{\mathcal{U}}
\newcommand{\controlSet}{\mathbb{U}}
\newcommand{\controlNorm}[1]{\|#1\|_{\controlSpace}}

\newcommand{\paramNorm}[1]{\|#1\|_{\paramSpace}}

\newcommand{\solOperator}{\mathcal{S}}
\newcommand{\romOperator}{\mathcal{S}_n}
\newcommand{\hankelOperator}{\mathcal{H}}

\newcommand{\reachGramian}{\mathcal{P}}
\newcommand{\observGramian}{\mathcal{Q}}

\def\Lt{{L^2}}

\newcommand{\activeSubspaceDistance}[1]{d_{\mathrm{A}}^{#1}}

\newcommand*{\rechterWinkel}[3]{
\draw[shift={(#2:#3)}] (#1) arc[start angle=#2, delta angle=90, radius = #3];
\fill[shift={(#2+45:#3/2)}] (#1) circle[radius=1.25\pgflinewidth];
}

\newcommand{\mail}[1]{\href{mailto:#1}{#1}}

\title{Kolmogorov $n$-widths for linear dynamical systems}
\author{Benjamin Unger\thanks{Institut für Mathematik, Technische Universität Berlin, Straße des 17. Juni 136, 10623 Berlin, Germany, \mail{unger@math.tu-berlin.de}. The work of this author was funded by the DFG Collaborative Research Center 910 \emph{Control of self-organizing nonlinear systems: Theoretical methods and concepts of application}.} \and Serkan Gugercin\thanks{Department of Mathematics, Virginia Polytechnic Institute and State University, Blacksburg, VA 24061, USA, \mail{gugercin@vt.edu}. The work of this author was supported in parts by NSF through Grant DMS-1720257 and by the Alexander von Humboldt Foundation.}}

\begin{document}

\maketitle

\begin{abstract}
	Kolmogorov $n$-widths and Hankel singular values are two commonly used concepts in model reduction. Here we show that for the special case of linear time-invariant dynamical (LTI) systems, these two concepts are directly connected. More specifically, 
the greedy search applied to the Hankel operator of an LTI system resembles the minimizing subspace for the Kolmogorov n-width and	the  Kolmogorov $n$-width of an LTI system equals its  $(n+1)st$ Hankel singular value once the subspaces are appropriately defined. We also establish a lower bound for the Kolmorogov $n$-width for parametric LTI systems and illustrate that the method of active subspaces can be viewed as the dual concept to the minimizing subspace for the Kolmogorov $n$-width.
\end{abstract}
\vskip .3truecm
\noindent
{\bf AMS(MOS) subject classification:} 37M99, 47B35, 65P99, 34A45, 35A35, 93C05

\renewcommand{\thefootnote}{\fnsymbol{footnote}}

\renewcommand{\thefootnote}{\arabic{footnote}}

\section{Introduction}
\label{sec:introduction}
The model reduction research has made great progress over the last two decades with major developments in many aspects ranging from linear to nonlinear and to parametric models, to data-driven model reduction and more. The resulting theory and algorithms were successfully applied to various applications, ranging from inverse problems to shape optimization to uncertainty quantification.  We refer the reader to the recent surveys and books \cite{BenCOW17,BauBF14,QuaMN16,Ant05,benner2015survey,hesthaven2016certified,chiHRW16} for further details.

Due to its wide range of applications, the model reduction research is carried out by a diverse community, at times different groups using their own tools and language to describe similar mathematical quantities. 

The Hankel singular values, heavily used in the systems and control theory community, and the Kolmogorov $n$-widths, heavily used in the reduced basis community, are two fundamental concepts  in model reduction. 
A connection between the two has been pointed out in \cite{Djo08,Djo10}. Even though in the earlier work \cite{Djo08}  the subspaces were properly identified,  the subspace assumptions in the latter paper \cite{Djo10}  seem to lead to a contradicting conclusion 
(see \Cref{rem:wrongStatement}); thus requiring further inspection. Therefore, we present 
the Hankel singular values and the Kolmogorov $n$-width connection in a self-contained manner, by detailing the underlying subspaces and by using a different proof. 
Further contributions are the following:
\begin{itemize}
	\item We show that optimal Hankel norm approximation yields a reduced system that is optimal in the sense that it attains the Kolmogorov $n$-width (\Cref{cor:HankelNormApprox}).
	\item \Cref{thm:activeSubspace} and \Cref{rem:activeSubspace} illustrate that for linear time invariant systems the method of \emph{active subspaces} \cite{Con15} can be understood as the dual concept to the minimizing subspace for the Kolmogorov $n$-width.
	\item We give a lower bound for the Kolmorogov $n$-width for parametric linear time invariant system in \Cref{thm:KolmogorovParametric}.
\end{itemize}

\section{Problem Setting}
\label{sec:problemSetting}

 Given a closed subset $\inpSet$ of a Hilbert space $(\inpSpace,\inpProd{\cdot}{\cdot})$, called the \emph{set of admissible inputs} and the \emph{input space}, respectively, and a further Hilbert space $(\outSpace,\outProd{\cdot}{\cdot})$, called the \emph{output space}, the (physical) system under investigation is described by an operator  $\solOperator$
\begin{equation}
	\label{eq:solOperator}
	\solOperator\colon\inpSet\to\outSpace,\qquad \inpVar\mapsto\outVar = \solOperator(\inpVar),
\end{equation}
where $\inpVar\in\inpSet\subseteq\inpSpace$ and $\outVar\in\outSpace$ denote, respectively, the \emph{inputs} and $\emph{outputs}$ of the operator $\solOperator$.
Our standing assumption is that we are interested in evaluating $\solOperator$ for many input values $\inpVar$ and that the evaluation of $\solOperator(\inpVar)$ is (computationally) demanding. Therefore, we would like to approximate $\solOperator$ with a surrogate operator 
\begin{equation}
	\label{eq:romSolOperator}
	\romOperator\colon\inpSet\to\outSpace, \qquad \inpVar\mapsto \redOutVar = \romOperator(\inpVar),
\end{equation}
such that $\outNorm{\outVar-\redOutVar} = \outNorm{(\solOperator-\romOperator)(\inpVar)}$ is small for all $\inpVar\in\inpSet$ and
that evaluations of $\romOperator$ are computationally cheaper. Hereby, $\outNorm{\cdot}$ denotes the norm induced by the inner product $\outProd{\cdot}{\cdot}$. In many applications the operator $\solOperator$ in \eqref{eq:solOperator} is given implicitly in terms of a (partial) differential-algebraic equation of the form
\begin{equation}
	\label{eq:genDynSys}
	\left\{\begin{aligned}
		0 &= F(t,\state(t),\dot{\state}(t),\param,\control(t)),\\
		\outVar(t) &= g(t,\state(t),\param,\control(t)),\\
		\state(t_0) &= \state_0,
	\end{aligned}\right.
\end{equation}
where the \emph{state variable} $\state$ is, for all $t$, an element of some Banach space $(\stateSpace,\stateNorm{\cdot})$ -- called the \emph{state space} -- and we use the convention $\dot{\state} \vcentcolon= \mathrm{d}/\mathrm{d}t\,\state$. The variable $\state_0\in\stateSpace$ is called \emph{initial condition} and the input space $\inpSpace$ is separated into a (time-independent) parameter space $\paramSpace$ and a control space $\controlSpace$ (i.\,e. $\inpSpace = \paramSpace\times\controlSpace$). Notice that in some cases the state $\state$ itself is of interest, in which case one can use as output function the identity on the state space, i.\,e. 
\begin{displaymath}
	g:\mathbb{R}\times\stateSpace\times\paramSpace\times\controlSpace\to\stateSpace,\qquad (t,\state,\param,\control) \mapsto \state.
\end{displaymath}
A special case of \eqref{eq:genDynSys} is \emph{parametric linear time-invariant} (pLTI) dynamical systems of the form
\begin{equation}
	\label{eq:pLTI}
	\system(\param)\colon\qquad \left\{
	\begin{aligned}
		\dot{\state}(t) &= A(\param)\state(t) + B(\param)\control(t),\\
		\outVar(t) &= C(\param)\state(t) + D(\param)\control(t),\\
		\state(0) &= \state_0,
	\end{aligned}
	\right.
\end{equation}
where $A\colon\paramSpace\to\mathbb{R}^{\stateDim\times\stateDim}$, $B\colon\paramSpace\to\mathbb{R}^{\stateDim\times\controlDim}$, $C\colon\paramSpace\to\mathbb{R}^{\outputDim\times\stateDim}$, and $D\colon\paramSpace\to\mathbb{R}^{\outputDim\times\controlDim}$ are smooth functions. In principle, $A(\param),B(\param),C(\param)$, and $D(\param)$ could  be operators on infinite dimensional spaces. In practice, however, model reduction usually starts with a  finite dimensional, albeit large-scale state-space, which is usually obtained by a semi-discretization of the infinite dimensional space.  The resulting large-scale finite-dimensional system \eqref{eq:pLTI} is often referred to as the \emph{truth} model. We assume that the truth model is accurate enough and that
its approximation error is negligible with respect to the model reduction error to follow. If $A,B,C,D$ in \eqref{eq:pLTI} are constant, i.e., independent of $\param$, we call \eqref{eq:pLTI} a \emph{linear time-invariant} (LTI) system.

A unifying feature of many model reduction schemes is that they can be formulated in the projection framework. Hence the construction of $\romOperator$ is mainly based on identifying a smaller linear subspace $\outSpaceRed\subseteq \outSpace$, that is a good approximation to $\solOperator(\inpSet)$. To simplify the notation, we write $\outSpaceRed\leq\outSpace$ to denote that $\outSpaceRed$ is a subspace of $\outSpace$. The subspace
$\outSpaceRed$  is called the (parametrized) \emph{manifolds of solutions} \cite{LasMQR14}. Subsequently, the operator $\solOperator$ is projected onto this linear subspace. A natural question to ask is what the \emph{best/optimal} subspace of a given dimension is, where the optimality is quantified by means of the minimal worst-case approximation error. Mathematically, this is described by the notion of so-called \emph{Kolmogorov $n$-widths} \cite{Kol36}, denoted by $d_n(\solOperator(\inpSet))$:
\begin{displaymath}
	d_n(\solOperator(\inpSet)) \vcentcolon= \inf_{\substack{\outSpaceRed\leq \outSpace\\\dim(\outSpaceRed)\leq n}} d(\outSpaceRed, \solOperator(\inpSet)),
\end{displaymath}
where $d(\outSpaceRed, \solOperator(\inpSet))$ is the largest distance between any point in $\solOperator(\inpSet)$ and the subspace $\outSpaceRed$, defined as
\begin{displaymath}
	 d(\outSpaceRed, \solOperator(\inpSet)) \vcentcolon= \sup_{\outVar\in\solOperator(\inpSet)} \inf_{\redOutVar \in \outSpaceRed} \outNorm{\outVar-\redOutVar}.
\end{displaymath}
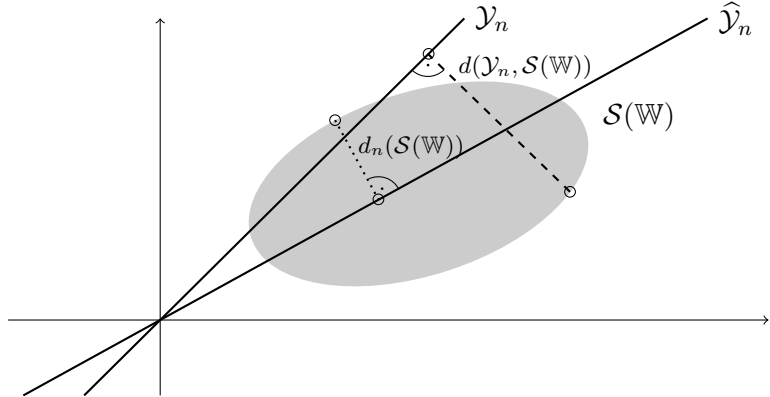
\begin{figure}[ht]
	\centering
	\begin{tikzpicture}
		\draw[->] (0,1) -- (10,1); 
		\draw[->] (2,0) -- (2,5); 
		
		\draw[thick,rotate=18,fill=gray!40,gray!40] (6,1) ellipse (2.3 and 1.2);
		\node at (8.3,3.7) {$\solOperator(\inpSet)$};
		
		\draw[thick] (1,0) -- (6,5) node[anchor=west] {$\outSpaceRed$};
		\draw[dashed, thick] (7.39,2.7) -- (5.52,4.53);
		\rechterWinkel{5.52,4.53}{226}{0.3}
		\draw (7.39,2.7) circle (2pt);
		\draw (5.53,4.53) circle (2pt);

		\draw[thick] (0.2,0) -- (9.2,5) node[anchor=west] {$\outSpaceRedMin$};
		\draw[dotted, thick] (4.3,3.65) -- (4.87,2.6);
		\rechterWinkel{4.87,2.6}{26}{0.3}
		\draw (4.3,3.65) circle (2pt);
		\draw (4.87,2.6) circle (2pt);	
		
		\node at (5.3,3.3) {\footnotesize $d_n(\solOperator(\inpSet))$};
		\node at (6.8,4.35) {\footnotesize $d(\outSpaceRed, \solOperator(\inpSet))$};
	\end{tikzpicture}
	\caption{Shematic illustration of the Kolmogorov $n$-width}
	\label{fig:kolmogorov}
\end{figure}%
The situation is illustrated in \Cref{fig:kolmogorov}. The length of the dashed line represents the distance from $\outSpaceRed$ to $\solOperator(\inpSet)$. Minimizing that distance over all supspaces results in the dotted line, whose length is the Kolmogorov $n$-width. Notice that we have used orthogonal projections onto $\outSpaceRed$, a technique that we use later in our derivations as well. Indeed, in a Hilbert space setting the Kolmogorov $n$-width can be equivalently formulated via linear projectors and there always exists a \emph{minimizing subspace} $\outSpaceRedMin$, i.\,e., we have
\begin{displaymath}
	d_n(\solOperator(\inpSet)) = d(\outSpaceRedMin, \solOperator(\inpSet)).
\end{displaymath} 
This result and relations to further $n$-widths are presented in the monograph \cite{Pin85}. In addition, for special classes
of problems, one can show that the Kolmogorov $n$-widths decay exponentially \cite{MadPT02,MadPT02b,QuaMN16}, thus enabling model reduction  to succeed.

The question at hand is how to construct spaces $\outSpaceRed$ such that $d_n(\solOperator(\inpSet)) \approx d(\outSpaceRed, \solOperator(\inpSet))$. A standard approach employed in the reduced basis community is the greedy construction \cite{PruRVMPT2002}, which iteratively enlarges the space $\outSpaceRed$ such that the worst approximation error is minimized. More precisely, let 
\begin{displaymath}
	e(\solOperator,\outSpaceRed,\inpSet) \vcentcolon= \sup_{\inpVar\in\inpSet} \outNorm{(\solOperator - \pi_{\outSpaceRed}\circ \solOperator)(\inpVar)}
\end{displaymath}
denote the worst approximation error for the projection of $\solOperator$ onto $\outSpaceRed$ where $\pi_{\outSpaceRed}$ denotes the orthogonal projection onto $\outSpaceRed$. We call a sequence $(\outVar_1,\ldots,\outVar_n)\in\outSpace^n$ a \emph{greedy sequence} if it satisfies
\begin{equation}
	\label{eq:greedySequence}
	\inf_{\boldsymbol{\phi}_i\in\outSpace} e(\solOperator,\spann\{\outVar_1,\ldots,\outVar_{i-1},\boldsymbol{\phi}_i\},\inpSet) = e(\solOperator,\spann\{\outVar_1,\ldots,\outVar_i\},\inpSet),
\end{equation}
for $i=1,\ldots,n$. In a practical implementation, the true error $e$ must be replaced by a cheap-to-evaluate error estimator. Morever $\inpSet$ is replaced by a discrete sampling of $\inpSet$ \cite{QuaRM11,QuaMN16} or the greedy search can be formulated as a sequence of adaptive model-constrained
optimization problems \cite{BuiThanh2008}. It is clear that such a (weak) greedy search might not be optimal. 
However, for special cases it is proven that the error for the subspace generated spanned by a greedy sequence converges asymptotically to the Kolmogorov $n$-widths \cite{Binev_etal2011}. 
One of our main results here, mainly  \Cref{thm:KolmogorovHankelGreedy}, illustrates that the greedy search applied to the Hankel operator for an LTI system resembles the minimizing subspace for the Kolmogorov $n$-width.

\section{Linear Time-Invariant Dynamical Systems}
When the matrix functions in \eqref{eq:pLTI} are constant over the parameter space or if we are only interested in controlling system \eqref{eq:pLTI} for a given parameter, the underlying dynamics simplifies to the LTI system
\begin{equation}
	\label{eq:LTI}
	\system: \qquad \dot{\state}(t) = A\state(t) + B\control(t),\qquad \outVar = C\state(t) + D\control(t),
\end{equation}
with $A\in\mathbb{R}^{\stateDim\times\stateDim}$, $B\in\mathbb{R}^{\stateDim\times\controlDim}$, $C\in\mathbb{R}^{\outputDim\times\stateDim}$, and $D\in\mathbb{R}^{\outputDim\times\controlDim}$.
As commonly done in the control literature in analyzing the input-to-output mapping, we assume a zero initial condition, i.\,e. $\state(0) = 0$. Further, we assume that the system $\system$ is asymptotically stable, i.e., all the eigenvalues of $A$ have negative real parts.
Then, the input-to-output mapping is given by the convolution integral
\begin{equation}
	\label{eq:convolutionOperator}
	\outVar(t) = (\solOperator\control)(t) \vcentcolon= \int_0^t h(t-s)\control(s)\mathrm{d}s,
\end{equation}
where $h(t) \vcentcolon= C\exp(tA)B + D\delta(t)$ is the \emph{impulse response} of the system and $\delta$ denotes the Dirac impulse. 

Before we compute the Kolmogorov $n$-widths in this setting, we make the following observations. If $\inpSet$ is a subspace of $\inpSpace$, then $\solOperator(\inpSet)$ is a subspace of $\outSpace$, and hence the Kolmogorov $n$-widths are either zero or infinity, thus giving no valuable information. Therefore, we need to assume that the set of admissible inputs $\inpSet$ is bounded, which results in $\solOperator(\inpSet)$ being bounded. Moreover, Proposition~1.2 in \cite{Pin85} states that in this setting, $d_n(\solOperator(\inpSet))$ converges to zero if and only if the closure of $\solOperator(\inpSet)$ is compact, thus we need to assume that $\solOperator$ is compact. Unfortunately, the convolution operator in \eqref{eq:convolutionOperator} is not compact in general \cite{Ant05}. This issue is resolved by modifying the domain and co-domain of the convolution operator to obtain the so-called \emph{Hankel operator} $\hankelOperator$ that maps past inputs to future outputs. More precisely, we have
\begin{equation}
	\label{eq:HankelOperator}	
	\hankelOperator\colon\Lt(-\infty,0;\mathbb{R}^{\controlDim})\to\Lt(0,\infty;\mathbb{R}^{\outputDim}),\qquad (\hankelOperator \control)(t) = \int_{-\infty}^0 h(t-s)\control(s)\mathrm{d}s.
\end{equation}
The Hankel operator $\hankelOperator$ is a finite-rank operator of at most rank $N$, and thus in particular compact \cite{Ant05}. The singular values of $\mathcal{H}$ can be computed as the square roots of the eigenvalues of the products of the Gramians $\reachGramian\observGramian$, which solve the Lyapunov equations
\begin{equation}
	\label{eq:Lyapunov}
	A\reachGramian + \reachGramian A^T + BB^T = 0 \qquad\text{and}\qquad A^T\observGramian + \observGramian A + C^TC = 0. 
\end{equation}
The singular values of the Hankel operator, called the \emph{Hankel singular values},
 play a fundamental role in control theory, especially in model reduction, see \cite{Ant05}.
  We denote the $i$th Hankel singular value of the system $\system$ by $\sigma_i(\system)$ with the convention $\sigma_i(\system) \geq \sigma_{i+1}(\system)$ for $i=1,\ldots,N-1$. The \emph{Hankel-norm} $\|\cdot\|_H$ of $\system$ is the $\Lt-\Lt$ induced norm 
  of the Hankel operator $\hankelOperator$ and one can show (cf. \cite{Ant05}) that it equals the largest Hankel singular value:
  \begin{equation}
	\label{eq:HankelNorm}
	\|\system\|_H = \sigma_1(\system).
\end{equation}
Since the Hankel operator is compact, it possesses a singular value decomposition (SVD) of the form
\begin{displaymath}
	\hankelOperator\control = \sum_{i=1}^\stateDim \sigma_i(\system)\inpProd{\control}{\bff_i}\,\bfg_i.
\end{displaymath}
where the orthonormal sets $\{\bff_i\}$ and $\{\bfg_i\}$ can be computed explicitly in terms of the eigenfunctions of $\reachGramian\observGramian$, see \cite{Ant05} for further details. Our main result  establishes the connection between the SVD of the Hankel Operator, the Kolmogorov $n$-widths, and the greedy search.
\begin{theorem}
	\label{thm:KolmogorovHankelGreedy}
	Let $\system = (A,B,C,D)$ be an asymptotically stable dynamical system with the Hankel operator $\hankelOperator = \sum_{i=1}^\stateDim \sigma_i(\system)\inpProd{\cdot}{\bff_i}\bfg_i$, $\inpSpace \vcentcolon= \Lt(-\infty,0;\mathbb{R}^{\controlDim})$ with standard inner product $\inpProd{\cdot}{\cdot}$, and let
	\begin{displaymath}
		\inpSet\vcentcolon= \{\control\in\inpSpace \mid \inpNorm{\control}\leq 1\}
	\end{displaymath}
	be the unit ball in the input space $\inpSpace$. Then $(\bfg_1,\ldots,\bfg_\stateDim)$ is a greedy sequence and 
	\begin{equation}
		\label{eq:KolmogorovHankel}
		d_n(\hankelOperator(\inpSet)) = d(\spann\{\bfg_1,\ldots,\bfg_n\},\hankelOperator(\inpSet)) = \sigma_{n+1}(\system),
	\end{equation}
	for $n=1,\ldots,N$.
\end{theorem}
\begin{proof}
	Let $\outSpaceRed$ denote an $n$-dimensional subspace of $\outSpace$. Since $L^2(0,\infty,\mathbb{R}^{\outputDim})$ is a Hilbert space, we have for $\outVar\in\outSpace$
	\begin{displaymath}
		\inf_{\redOutVar \in \outSpaceRed} \outNorm{\outVar-\redOutVar} = \outNorm{\outVar - \pi_{\outSpaceRed}\outVar},
	\end{displaymath}		
	where $\pi_{\outSpaceRed}$ is the orthogonal projection onto $\outSpaceRed$. Thus
	\begin{displaymath}
		d_n(\hankelOperator(\inpSet)) = \inf_{\substack{\outSpaceRed\leq \outSpace\\\dim(\outSpaceRed)=n}} \sup_{\outVar\in\solOperator(\inpSet)} \outNorm{\outVar-\pi_{\outSpaceRed}\outVar} =  \inf_{\substack{\outSpaceRed\leq \outSpace\\\dim(\outSpaceRed)=n}} e(\hankelOperator,\outSpaceRed,\inpSet),
	\end{displaymath}
	which shows (since $\hankelOperator$ is linear) that the minimizing subspace for the Kolmogorov $n$-width and the subspace generated by the greedy sequence in \eqref{eq:greedySequence} coincide. 
	
	It remains to show that $d_n(\hankelOperator(\inpSet)) = \sigma_{n+1}$. For general compact operators a proof of this fact is given in \cite{Pin85}, which invokes a primal-dual approach and the Courant-Fischer-Weyl min-max principle for self-adjoint compact operators \cite{ReeS78}. 	Here, to explicitly highlight how the input and output spaces appear and to be able to use it later in proving \Cref{thm:activeSubspace}, we give a version, which follows the Schmidt-Eckart-Young-Mirsky theorem \cite{Ant05} for optimal low-rank approximation in the finite dimensional case.
	
	Recall that in out case the Hankel operator is a finite rank operator. Let $\mathcal{F} \vcentcolon= \spann\{\bff_1,\ldots,\bff_N\}$ and $\mathcal{F}^\perp$ be its orthogonal complement such that $\inpSpace = \mathcal{F}\oplus\mathcal{F}^\perp$. Notice that we have $\hankelOperator(\mathcal{F}^\perp) = \{0\}$. Then $\dim(\pi_{\outSpaceRed}\hankelOperator(\inpSet)) \leq n$ and  there exists 
	\begin{equation}
		\label{eq:nWidthLowerBoundInput}
		\inpVar\in\ker(\pi_{\outSpaceRed}\hankelOperator)\cap \spann\{\bff_1,\ldots,\bff_{n+1}\}
	\end{equation}
	with $\inpNorm{\inpVar} = 1$; thus $\inpVar\in\inpSet$. We then obtain
	\begin{align*}
		e(\hankelOperator,\outSpaceRed,\inpSet)^2 &\geq \outNorm{(\hankelOperator-\pi_{\outSpaceRed}\circ\hankelOperator)\inpVar}^2 = \outNorm{\hankelOperator\inpVar}^2 = \outNorm{\sum_{i=1}^{n+1} \sigma_i(\system) \inpProd{\inpVar}{\bff_i}\bfg_i}^2\\
		&\geq \sigma_{n+1}(\system)^2\outNorm{\sum_{i=1}^{n+1} \inpProd{\inpVar}{\bff_i} \bfg_i}^2 = \sigma_{n+1}(\system)^2 \sum_{i=1}^{n+1} |\inpProd{\inpVar}{\bff_i}|^2\\
		&= \sigma_{n+1}(\system)^2,
	\end{align*}
	yielding $d_n(\hankelOperator(\inpSet))\geq \sigma_{n+1}(\system)$. On the other hand, the choice $\outSpaceRed = \spann\{\bfg_1,\ldots,\bfg_n\}$ yields 
	\begin{displaymath}
		\pi_{\outSpaceRed}\hankelOperator = \sum_{i=1}^n \sigma_i(\system)\inpProd{\cdot}{\bff_i}\bfg_i
	\end{displaymath}
	and hence 
	\begin{align*}
		e(\hankelOperator,\spann\{\bfg_1,\ldots,\bfg_n\},\inpSet)^2 &= \sup_{\inpVar\in\inpSet} \outNorm{(\hankelOperator - \pi_{\outSpaceRed}\circ\hankelOperator)\inpVar}^2 = \sup_{\inpVar\in\inpSet} \outNorm{\sum_{i=n+1}^\stateDim \sigma_i(\system)\inpProd{\inpVar}{\bff_i}\bfg_i}^2\\
		&= \sup_{\inpVar\in\inpSet} \sum_{i=n+1}^\stateDim \sigma_i(\system)^2 |\inpProd{\inpVar}{\bff_i}|^2.
	\end{align*}
	Using
	\begin{align*}
		\sum_{i=n+1}^\stateDim \sigma_i(\system)^2 |\inpProd{\inpVar}{\bff_i}|^2 &\leq \sigma_{n+1}(\system)^2 \sup_{\inpVar\in\inpSet} \sum_{i=n+1}^\stateDim |\inpProd{\inpVar}{\bff_i}|^2 \leq \sigma_{n+1}(\system)^2
	\end{align*}
	with equality for $\inpVar = \bff_{n+1}$, we obtain $e(\hankelOperator,\spann\{\bfg_1,\ldots,\bfg_n\},\inpSet) = \sigma_{n+1}$. Thus we have
	\begin{displaymath}
		d_n(\hankelOperator(\inpSet)) = \sigma_{n+1}(\system),
	\end{displaymath}
	which completes the proof.
\end{proof}
\begin{remark}
	\label{rem:wrongStatement}
	\Cref{thm:KolmogorovHankelGreedy} seems to contradict a result from \cite{Djo10}, where the author claims (cf. \cite[Theorem~1]{Djo10}) that 
	\begin{equation}
		\label{eq:KolHankel}
		d_n(\hankelOperator(\Lt(-\infty,0;\mathbb{R}^{\controlDim}))) = \sigma_{n+1}(\system)
	\end{equation}
	for an asymptotically stable LTI system $\system$. However, since the Hankel operator $\hankelOperator$ is linear, the set $\hankelOperator(\Lt(-\infty,0;\mathbb{R}^{\controlDim}))$ is a linear subspace and thus the Kolmogorov $n$-widths are either infinity or zero, which shows that \eqref{eq:KolHankel} cannot be true.
\end{remark}

Note that even with the knowledge of the minimizing subspace $\outSpaceRedMin\vcentcolon=\spann\{\bfg_1,\ldots,\bfg_n\}$, it is a nontrivial  task to determine a state-space representation $\system_n = (A_n,B_n,C_n,D_n)$ of $\pi_{\outSpaceRedMin}\hankelOperator$. This  issue was resolved by Glover \cite{Glo84} who developed a computational procedure for constructing $\pi_{\outSpaceRedMin}\hankelOperator$. Together with the Adamjan-Arov-Krein theorem \cite{AdaAK71}, this implies the following result.
\begin{corollary}
	\label{cor:HankelNormApprox}
	Let the assumptions and definitions be as in \Cref{thm:KolmogorovHankelGreedy}. Then 
	\begin{equation}
		\label{eq:KolmogorovOptimalHankelNorm}
		d_n(\hankelOperator(\inpSet)) = \inf_{\substack{\tilde{\system}\,\mathrm{asym.\,stable}\\\dim(\tilde{\system})\leq n}} \|\system - \tilde{\system}\|_H.
	\end{equation}
\end{corollary}

\section{Connection to Active Subspaces}
Instead of looking at the best approximation of the image of $\solOperator$ in terms of a linear subspace, we can also ask for the best approximation of the input space of $\solOperator$ in terms of a low-dimensional subspace, which leads to the notion of the so called \emph{active subspace} (see \cite{Con15} and the references therein). In other words, the active subspace describes the important directions in the input space $\inpSpace$. More precisely, we call an $n$-dimensional subspace $\inpSpaceRedMin\leq\inpSpace$ an \emph{active subspace} if it satisfies
\begin{equation}
	\label{eq:activeSubspace}
		\activeSubspaceDistance{n}(\mathcal{S},\inpSet) \vcentcolon= \sup_{\inpVar\in\inpSet} \outNorm{\solOperator(\inpVar) - \solOperator(\pi_{\inpSpaceRedMin}\inpVar)} = \inf_{\substack{\inpSpaceRed\leq \inpSpace\\\dim(\inpSpaceRed)\leq n}} \sup_{\inpVar\in\inpSet} \outNorm{\solOperator(\inpVar) - \solOperator(\pi_{\inpSpaceRed}\inpVar)},
\end{equation}
where $\smash{\pi_{\inpSpaceRed}}$ denotes the orthogonal projection onto $\inpSpaceRed$.
Having identified an active subspace means that the computational cost of a parameter study in $\inpSet$ can be reduced by performing the parameter study in $\smash{\pi_{\inpSpaceRedMin}}\inpSet$. Note that similar to the minimizing subspace for Kolmogorov $n$-widths, in practial applications the minimization problem in \eqref{eq:activeSubspace} is not resolved exactly but only approximately \cite{Con15}. As for the Kolmogorov $n$-widths, the active subspace for the Hankel operator restricted to the unit ball can be computed exactly, as the following result shows.
\begin{theorem}
	\label{thm:activeSubspace}
	Let the assumptions and definitions be as in \Cref{thm:KolmogorovHankelGreedy}. Then the $n$-dimensional active subspace is given by $\inpSpaceRedMin = \spann\{\bff_1,\ldots,\bff_n\}$ with worst-case approximation error $\activeSubspaceDistance{n}(\hankelOperator,\inpSet) = \sigma_{n+1}(\system)$.
\end{theorem}
\begin{proof}
As can be expected, the proof follows similarly to that of  \Cref{thm:KolmogorovHankelGreedy}; thus we only give a brief sketch. Let $\inpSpaceRed$ be an $n$-dimensional subspace of $\inpSpace$ and let $\inpSpaceRed^\perp$ denote its orthogonal complement. Then there exists $\inpVar\in\inpSpaceRed^\perp \cap \spann\{\bff_1,\ldots,\bff_{n+1}\}$ with $\|\inpVar\|=1$, and thus, $\smash{\outNormSmashed{\hankelOperator(\inpVar - \pi_{\inpSpaceRed}\inpVar)}^2}  \geq \sigma_{n+1}(\system)^2$,
	which shows $\activeSubspaceDistance{n}(\mathcal{\hankelOperator},\inpSet) \geq \sigma_{n+1}(\system)$. Conversely, the choice $\inpSpaceRed = \spann\{\bff_1,\ldots,\bff_n\}$ yields
$		\smash{\outNormSmashed{\hankelOperator(\inpVar-\pi_{\inpSpaceRed}\inpVar)}^2}  \leq \sigma_{n+1}(\system)^2$
	with equality for $\inpVar = \bff_{n+1}$. Thus the active subspace is given by $\inpSpaceRedMin = \spann\{\bff_1,\ldots,\bff_n\}$ with approximation error $\activeSubspaceDistance{n}(\hankelOperator,\inpSet) = \sigma_{n+1}(\system)^2$.
\end{proof}
\begin{remark}
	\label{rem:activeSubspace}
	Since the vectors $\bff_i$'s can be computed as the minimizing subspace for \emph{the adjoint system} $\system^* = (-A^*,-C^*,B^*,D^*)$, we can interpret the active subspace as the dual concept of the minimizing subspace for the Kolmogorov $n$-width. Hence, the Hankel operator and the greedy selection procedure can be seen as the linking theory between Kolmogorov $n$-widths and active subspaces.
\end{remark}

\section{Parametric LTI systems}
As expected, the analysis for the pLTI case  in \eqref{eq:pLTI} is more involved than for the LTI case. Let $\controlSpace \vcentcolon= \Lt((-\infty,0],\mathbb{R}^{\controlDim})$ and assume that the control variable $\control$ is an element of the unit ball $\controlSet \vcentcolon=\{\control\in\controlSpace\mid \controlNorm{u}\leq 1\}$ and that the parameter $\param$ varies in a compact parameter set $\paramSet\subseteq\paramSpace$. In particular we set $\inpSet = \paramSet\times\controlSet$ with norm
\begin{displaymath}
	\inpNorm{(\param,\control)} = \paramNorm{\param} + \controlNorm{\control}.
\end{displaymath}
For each $\param\in\paramSet$, the Hankel operator $\mathcal{H}(\param)$ is given by
\begin{displaymath}
	(\hankelOperator(\param,\control))(t) = \int_{-\infty}^0 h(\param,t-s)\control(s)\mathrm{d}s \quad\text{with}\quad 
	h(\param,t) = C(\param)\exp(tA(\param))B(\param) + \delta(t)D(\param).
\end{displaymath}
We are interested in the Kolmogorov $n$-width $d_n(\hankelOperator(\paramSet,\controlSet))$. Since for the constant parameter case we needed to assume that the system is asymptotically stable, we assume that \eqref{eq:pLTI} is  asymptotically stable for each $\param\in\paramSet$. This set-up leads to our final result.
\begin{theorem}
	\label{thm:KolmogorovParametric}
	Let $\outSpace = \Lt((-\infty,0],\mathbb{R}^{\outputDim})$, $\controlSpace = \Lt([0,\infty),\mathbb{R}^{\controlDim})$, and consider the asymptotically stable parametric LTI system $\system(\param)$ as in \eqref{eq:pLTI} with the Hankel operator $\hankelOperator$. Assume that $\paramSet$ is compact, the Hankel singular values $\sigma_i(\system(\param))$ are continuous on $\paramSet$, and set $\controlSet\vcentcolon= \{\control\in\controlSpace \mid \controlNorm{\control}\leq 1\}$. Then
	\begin{equation}
		\label{eq:mainResult}
		d_n(\hankelOperator(\paramSet,\controlSet)) \geq \max_{\param\in\paramSet} \sigma_{n+1}(\system(\param)).
	\end{equation}
\end{theorem}
\begin{proof}
	\allowdisplaybreaks
	Let $\outSpaceRed$ denote a subspace of $\outSpace$ with dimension $n\in\mathbb{N}$ and $\outSpaceRed^\perp$ its orthogonal complement (with respect to the standard inner product in $\outSpace$ -- which we denote with $\outProd{\cdot}{\cdot}$). Fix $\param\in\paramSet$ and consider $\outVar\in\hankelOperator(\param,\inpSet)$, $\redOutVar\in\outSpaceRed$ and $\boldsymbol{x}\in\outSpaceRed^\perp$. Then we obtain
	\begin{align*}
		\outNorm{\outVar-\redOutVar}^2 &= \outNorm{\pi_{\outSpaceRed}\outVar-\redOutVar + \left(\id_\outSpace-\pi_{\outSpaceRed}\right)\outVar}^2 \geq \outNorm{\left(\mathrm{Id}_{\outSpace}-\pi_{\outSpaceRed}\right)\outVar}^2,\\
		\outProd{\outVar}{\boldsymbol{x}} &= \outProd{\left(\mathrm{Id}_{\outSpace}-\pi_{\outSpaceRed}\right)\outVar}{\boldsymbol{x}} \leq \outNorm{\left(\mathrm{Id}_{\outSpace}-\pi_{\outSpaceRed}\right)\outVar} \outNorm{\boldsymbol{x}},
	\end{align*}
	and hence
	\begin{displaymath}
		\inf_{\redOutVar\in\outSpaceRed} \outNorm{\outVar-\redOutVar} = \sup_{\boldsymbol{x}\in\outSpaceRed^\perp\setminus\{0\}} \frac{\outProd{\outVar}{\boldsymbol{x}}}{\outNorm{\boldsymbol{x}}}.
	\end{displaymath}
	Thus we can reformulate the Kolmogorov $n$-width as
	\begin{align*}
		d_n(\hankelOperator(\paramSet,\inpSet)) &= \inf_{\substack{\outSpaceRed\leq \outSpace\\\dim(\outSpaceRed)=n}}\quad \sup_{\outVar\in\hankelOperator(\paramSet,\inpSet)}\quad \inf_{\redOutVar \in \outSpaceRed} \outNorm{\outVar-\redOutVar}\\
		&= \inf_{\substack{\outSpaceRed\leq \outSpace\\\dim(\outSpaceRed)=n}}\quad \sup_{\boldsymbol{x}\in\outSpaceRed^\perp\setminus\{0\}}\quad \sup_{\outVar\in\hankelOperator(\paramSet,\inpSet)} \frac{\outProd{\outVar}{\boldsymbol{x}}}{\outNorm{\boldsymbol{x}}}\\
		&= \inf_{\substack{\outSpaceRed\leq \outSpace\\\dim(\outSpaceRed)=n}}\quad \sup_{\boldsymbol{x}\in\outSpaceRed^\perp\setminus\{0\}}\quad \sup_{(\param,\control)\in\paramSet\times\inpSet} \frac{\outProd{\hankelOperator(\param)\control}{\boldsymbol{x}}}{\outNorm{\boldsymbol{x}}}\\
		&= \inf_{\substack{\outSpaceRed\leq \outSpace\\\dim(\outSpaceRed)=n}}\quad \sup_{\boldsymbol{x}\in\outSpaceRed^\perp\setminus\{0\}}\quad \sup_{(\param,\control)\in\paramSet\times\inpSet} \frac{\inpProd{\control}{\hankelOperator(\param)^*\boldsymbol{x}}}{\outNorm{\boldsymbol{x}}},
	\end{align*}
	where $\hankelOperator(\param)^*$ denotes the adjoint operator of $\hankelOperator(\param)$. Due to the definition of $\inpSet$ we have
	\begin{displaymath}
		\sup_{\control\in\inpSet} \frac{\inpProd{\control}{\hankelOperator(\param)^*\boldsymbol{x}}}{\outNorm{\boldsymbol{x}}} = \frac{\inpNorm{\hankelOperator(\param)^*\boldsymbol{x}}}{\outNorm{\boldsymbol{x}}} = \frac{\sqrt{\outProd{\hankelOperator(\param)\hankelOperator(\param)^*\boldsymbol{x}}{\boldsymbol{x}}}}{\outNorm{\boldsymbol{x}}}
	\end{displaymath}
	and thus
	\begin{align*}
		d_n(\hankelOperator(\paramSet,\inpSet)) &= \inf_{\substack{\outSpaceRed\leq \outSpace\\\dim(\outSpaceRed)=n}}\quad \sup_{\param\in\paramSet}\quad \sup_{\boldsymbol{x}\in\outSpaceRed^\perp\setminus\{0\}} \frac{\sqrt{\outProd{\hankelOperator(\param)\hankelOperator(\param)^*\boldsymbol{x}}{\boldsymbol{x}}}}{\outNorm{\boldsymbol{x}}}\\
		&\geq \sup_{\param\in\paramSet}\quad \inf_{\substack{\outSpaceRed\leq \outSpace\\\dim(\outSpaceRed)=n}}\quad  \sup_{\boldsymbol{x}\in\outSpaceRed^\perp\setminus\{0\}} \frac{\sqrt{\outProd{\hankelOperator(\param)\hankelOperator(\param)^*\boldsymbol{x}}{\boldsymbol{x}}}}{\outNorm{\boldsymbol{x}}}.
	\end{align*}
	Since $\hankelOperator(\param)$ is compact for every $\param$, so is $\hankelOperator(\param)\hankelOperator(\param)^*$ and from the Courant-Fischer-Weyl min-max principle for self-adjoint compact operators (cf. \cite{ReeS78}), we obtain
	\begin{displaymath}
		d_n(\hankelOperator(\param,\inpSet)) \geq \sqrt{\lambda_{n+1}(\hankelOperator(\param)\hankelOperator(\param)^*)} = \sigma_{n+1}\left(\system(\param)\right).
	\end{displaymath}
\end{proof}
\begin{remark}
	The continuity of the Hankel singular values with respect to the parameter $\param$ can be guaranteed if we assume that $A,B,C$, and $D$ are holomorphic on the logarithmically convex Reinhardt domain $\paramSet$, see \cite{WitTKAS16} for further details. \Cref{thm:KolmogorovParametric} reveals that for parametric LTI dynamical systems if  only a  \emph{non-parametric} approximation basis, i.e., a global basis,  is employed, one can only obtain a lower bound for the Kolmogorov $n$-width. In general, to achieve this lower bound
one will need to use a parametrically varying basis. \end{remark}

\section{Summary and Outlook}
In this paper, we have illustrated a direct connection between the Hankel singular values and the Kolmogorov $n$-widths for LTI systems. For parametric LTI systems, the same analysis has lead to a lower bound for the Kolmorogov $n$-width. We also showed that the method of active subspaces can be considered as the dual concept to the minimizing subspace for the Kolmogorov $n$-width. Extensions of these results to more general cases such as nonlinear dynamical systems will be of interest. Also, it will be interesting to investigate further if this connection can lead to an approximate, but numerically more feasible, implementation of optimal Hankel norm approximation.

\bibliographystyle{plain}
\bibliography{HankelOperator}

\end{document}